\documentclass[letterpaper,11pt]{article}
\usepackage{amsmath,mathrsfs,amssymb,amsfonts,amsthm,mathtools}
\usepackage{physics,color}
\usepackage{graphicx}
\usepackage{subfigure}
\usepackage{tabularx} 
\usepackage[table]{xcolor}
\usepackage{tikz,tikz-cd}
\usetikzlibrary{shapes.geometric, arrows, positioning}
\usepackage{empheq}
\usepackage{bm}
\newtheorem{theorem}{Theorem}

\pdfoutput=1 
\usepackage{jheppub} 
\usepackage[T1]{fontenc} 

\newcommand{\nn}{\nonumber}

\def\sst#1{{\scriptscriptstyle #1}}
\def\oneone{\rlap 1\mkern4mu{\rm l}}

\hypersetup{
    colorlinks=true,
    linkcolor=blue,
    citecolor=blue,
    urlcolor=blue
}

\title{Hodge Duals in Spherical Compactifications}

\author[a]{Arash Azizi}

\affiliation[a]{{\it The Institute for Quantum Science and Engineering,
Texas A\&M University,\\ College Station, TX 77843, U.S.A.}}

\emailAdd{sazizi@tamu.edu}

\abstract{
We present a general formalism for computing the Hodge dual of differential forms in arbitrary dimensions subject to a spherical constraint. This problem arises naturally in Kaluza-Klein compactifications, where sphere reductions demand careful treatment of differential forms constrained to lie on embedded submanifolds. We derive an explicit expression for the Hodge dual of a $p$-form in the presence of such constraints and validate our general ansatz through illustrative examples in three dimensions, including both flat and diagonal metric backgrounds. The resulting framework offers a systematic and practical tool for handling constrained Hodge duals, with direct applications to consistent truncations in supergravity and string theory compactifications.
}

\begin{document} 
\maketitle
\flushbottom

%%%%%%%%%%%%%%%%%%%%%%%%%%%%%%%%%%%%%%%%%%%%%%%%%%%%%%%%
\section{Introduction}%%%%%%%%%%%%%%%%%%%%%%%%%%%%%%%%%
%%%%%%%%%%%%%%%%%%%%%%%%%%%%%%%%%%%%%%%%%%%%%%%%%%%%%%%%
Dimensional reduction, first introduced by Kaluza \cite{Kaluza:1921tu} and Klein \cite{Klein:1926tv}, has remained a cornerstone in connecting higher-dimensional theories to the observed four-dimensional universe. In modern contexts, it provides a bridge between string/M-theory and effective low-energy physics. The simplest implementations—such as compactifications on tori—lead to Abelian gauge symmetries and can be treated using group-theoretic techniques. However, more physically compelling reductions, particularly those involving spheres, are far more intricate. Compactification on $S^n$ yields non-Abelian gauge groups, with the isometry group $SO(n+1)$ of the sphere inducing a gauged symmetry in the lower-dimensional theory. This mechanism underlies many constructions of gauged supergravities.

The revival of interest in sphere reductions during the 1980s was largely driven by the development of eleven-dimensional supergravity and its compactification on $S^7$ \cite{deWit:1986oxb}, leading to maximal $SO(8)$ gauged supergravity in four dimensions. Similarly, compactifications on $S^5$ and $S^4$ played pivotal roles in understanding $\text{AdS}_5/\text{CFT}_4$ and $\text{AdS}_7/\text{CFT}_6$ correspondences, respectively \cite{Kim:1985ez, Nastase:1999cb}. Additionally, Salam and Sezgin made significant contributions through their study of chiral compactifications of $N=2$ Einstein-Maxwell supergravity in six dimensions on Minkowski space times $S^2$ \cite{Salam_Sezgin1984}. However, ensuring the consistency of such compactifications—i.e., that all solutions of the lower-dimensional theory uplift to full solutions of the higher-dimensional theory—remains a subtle and often unsolved problem. Fully consistent sphere reductions have been achieved only in certain highly symmetric settings \cite{Cvetic:2000dm, Cvetic:2000yp}.

One persistent technical difficulty in this context is the proper definition and computation of the Hodge dual of differential forms under the constraint $\sum_{i} x^i x^i = 1$, arising from the embedding of the compact manifold (typically a round sphere) in a higher-dimensional ambient space. This difficulty is especially relevant in the construction of Freund-Rubin-type ansätze \cite{Freund:1980xh}, in which fluxes threading the compact space are related via Hodge duality to spacetime components of field strengths. Understanding the precise form of such duals is essential for both classical solution building and quantum consistency.

Various works have addressed this problem in specific examples. The series of papers by Cvetič, Lü, and Pope (see, e.g., \cite{Cvetic:1999au, Cvetic:2000ah, Cvetic:2000tb, Cvetic:2003jy, Cvetic:2003xr}) employed ingenious parametrizations of the internal manifold to compute Hodge duals and formulate consistent truncations. Related efforts have provided partial results for the Hodge dual of differential forms in five- and six-dimensional settings \cite{Gibbons:2003gp}. Yet despite these successes, a general expression for the Hodge dual of an arbitrary $p$-form in such constrained geometries has not appeared in the literature.

In this work, we address this gap by analyzing the problem in increasing generality. We first study the Hodge dual of differential forms on the 2-sphere embedded in $\mathbb{R}^3$, using both intrinsic and extrinsic methods. This case allows for a transparent comparison between coordinate-based and orthonormal-frame approaches and reveals how the constraint modifies the standard duality structure. We then proceed to arbitrary dimension and general background metrics, introducing a normalized projection vector field that projects ambient-space forms onto the constrained hypersurface. This leads to a closed-form expression for the Hodge dual that is consistent with the fundamental double-dual identity and valid for all diagonal and non-diagonal metrics.

The structure of this paper is as follows. In Section~\ref{sec2}, we review standard results on Hodge duality and establish our notational and convention choices. We also present several auxiliary results that will be essential in subsequent sections. In Section~\ref{sec3}, we examine specific examples in three dimensions, which provide insight into the general structure. In Section~\ref{sec.non-cov}, we formulate a general ansatz for the Hodge dual in the presence of constraints and express the result in (\ref{ansatz.final}). Finally, we conclude in Section~\ref{con} with a discussion of physical applications and directions for future work.

%%%%%
%%%%%
%%%%%%%%%%%%%%%%%%%%%%%%%%%%%%%%%%%%%%%%%%%%%%%%%%%%%%%%%%%%%%%%%%%%
\section{The notations and preliminaries} \label{sec2}
%%%%%%%%%%%%%%%%%%%%%%%%%%%%%%%%%%%%%%%%%%%%%%%%%%%%%%%%%%%%%%%%%%%%
We follow the conventions and notation of~\cite{popekk}. For completeness, we summarize the main definitions here. Consider a general $p$-form $\omega_{\sst{(p)}}$ on a $D$-dimensional spacetime with $t$ time-like directions:
\begin{align} 
\omega_{\sst{(p)}} = \frac{1}{p!} \, \omega_{\mu_1 \mu_2 \cdots \mu_p} \,
dx^{\mu_1} \wedge dx^{\mu_2} \wedge \cdots \wedge dx^{\mu_p} \,.
\end{align}
The Hodge dual of a basis $p$-form is defined by
\begin{align}
* (dx^{\mu_1} \wedge \cdots \wedge dx^{\mu_p})
= \frac{1}{q!} \, \epsilon_{\nu_1 \cdots \nu_q}{}^{\mu_1 \cdots \mu_p} \,
dx^{\nu_1} \wedge \cdots \wedge dx^{\nu_q}, \label{dual}
\end{align}
where $q = D - p$. The Levi-Civita tensor $\epsilon_{\mu_1 \cdots \mu_D}$ is defined as
\begin{align}
\epsilon_{\mu_1 \cdots \mu_D} = \sqrt{|g|} \, \varepsilon_{\mu_1 \cdots \mu_D},
\end{align}
where $\varepsilon_{\mu_1 \cdots \mu_D}$ is the Levi-Civita symbol, or tensor density, which takes values $\pm 1$ or $0$ depending on whether the set of indices $(\mu_1, \cdots, \mu_D)$ is an even permutation, odd permutation, or not a permutation of $(0,1,\dots,D-1)$, respectively.

The fully contravariant version of the Levi-Civita symbol satisfies
\begin{align}
\varepsilon^{\mu_1 \cdots \mu_D} = (-1)^t \, \varepsilon_{\mu_1 \cdots \mu_D},
\end{align}
where $t$ is the number of time-like components in the metric. Note that $\varepsilon_{\mu_1 \cdots \mu_D}$ is not a true tensor, so its indices cannot be raised or lowered with the metric. However, the Levi-Civita tensor $\epsilon_{\mu_1 \cdots \mu_D}$ is a proper tensor, and its indices may be raised and lowered. Hence, the contravariant form of the Levi-Civita tensor is given by
\begin{align}
\epsilon^{\mu_1 \cdots \mu_D} = \frac{1}{\sqrt{|g|}} \, \varepsilon^{\mu_1 \cdots \mu_D}.
\end{align}

From this point forward, we adopt a convention in which the inverse metric is explicitly denoted by $g^{-1}$, rather than by raising indices. For example, instead of writing $A^i = g^{ij} A_j$, we write $g^{-1}_{ij} A_j$ to indicate the corresponding contravariant vector. Unless stated otherwise, Einstein summation over repeated indices is always assumed.

%%%%%%%%%%%%%%%%%%%%%%%%%%%%%%%%%%%%%%%%%%%%%%%%%%%%%%%%%%%%%%%%%%%%%%
\section{Simple examples of the Hodge dual with a constraint}  \label{sec3}
%%%%%%%%%%%%%%%%%%%%%%%%%%%%%%%%%%%%%%%%%%%%%%%%%%%%%%%%%%%%%%%%%%%%%%

\subsection{Euclidean metric in three dimensions}

We begin by studying the Hodge dual of differential forms subject to a constraint, focusing on the simplest nontrivial example: the unit 2-sphere embedded in $\mathbb{R}^3$. Consider the Euclidean metric in three dimensions:
\begin{align}
ds^2 = dx^2 + dy^2 + dz^2, \qquad \text{with the constraint} \qquad x^2 + y^2 + z^2 = 1.
\end{align}
Our goal is to compute $*dx$ under this constraint.

To proceed, we adopt the standard spherical coordinate parametrization:
\begin{align}
x = \sin\theta \cos\phi, \qquad y = \sin\theta \sin\phi, \qquad z = \cos\theta.
\end{align}
Using the orthonormal frame (vielbein) formalism, the induced metric on the unit sphere becomes
\begin{align}
ds^2 = d\theta^2 + \sin^2\theta \, d\phi^2 = (e^1)^2 + (e^2)^2,
\end{align}
with the corresponding vielbeins
\begin{align}
e^1 = d\theta, \qquad e^2 = \sin\theta \, d\phi.
\end{align}
In this orthonormal frame, the Hodge dual acts as
\begin{align}
* e^1 = \epsilon_2{}^1 e^2 = \epsilon_{21} e^2 = -e^2, \qquad
* e^2 = \epsilon_1{}^2 e^1 = \epsilon_{12} e^1 = e^1,
\end{align}
with the convention $\epsilon_{12} = -\epsilon_{21} = 1$. Therefore, the duals of the coordinate one-forms are
\begin{align}
* e^1 = * d\theta = -e^2 = -\sin\theta \, d\phi, \qquad
* e^2 = * (\sin\theta \, d\phi) = e^1 = d\theta,
\end{align}
which immediately implies
\begin{align}
* d\phi = \frac{d\theta}{\sin\theta}. \label{star-viel}
\end{align}
Then we have
%%%%%%%%%%
\begin{align}
* dx &= \cos\theta \cos\phi \, *d\theta - \sin\theta \sin\phi \, *d\phi \nn\\
%%%%
&= \frac{1}{x^2 + y^2} \left( x y z \, dx - x^2 z \, dy + y \, dz \right). \label{stardx2}
\end{align}
%%%%%%
Furthermore, 
\begin{align}
*dy &= \cos\theta \sin\phi \, *d\theta + \sin\theta \cos\phi \, *d\phi \nn\\
%%%%
&= \frac{1}{x^2 + y^2} \left( y^2 z \, dx - x y z \, dy - x \, dz \right). \label{stardy}
\end{align}
%%%%%
Finally, the Hodge dual of $dz$ follows from
\begin{align}
*dz = -\sin\theta \, *d\theta = -\sin\theta (-\sin\theta \, d\phi)
= x \, dy - y \, dx. \label{stardz}
\end{align}
%%%%%%
To express the result just in terms of $dx$ and $dy$, one may eliminate $dz$ via $xdx + ydy + zdz = 0$. Then one finds:
\begin{align}
* dx =& \frac{1}{x^2 + y^2} (xyzdx - x^2zdy - \frac{y}{z}xdx - \frac{y}{z}ydy)
=  -\frac{xy}{z} dx - \frac{x^2z^2 + y^2}{(x^2 + y^2)z} dy,  \nn \\
%%%%%%%%
* dy =& \frac{1}{x^2 + y^2} (y^2zdx - xyzdy + \frac{x}{z}xdx + \frac{x}{z}ydy)
= \frac{(y^2z^2 + x^2)dx}{(x^2 + y^2)z}  + \frac{xy}{z}dy. \label{stardx.dy.dz}
\end{align}
%%%%%%%%
Note that the expressions in~\eqref{stardz}, and \eqref{stardx.dy.dz} are not manifestly symmetric. To obtain a more symmetric form, we eliminate $dx$ using the identity $xdx + ydy + zdz = 0$ 
and substitute this into $*dx$. Then $*dx$ and $*dy$ become
\begin{align}
* dx &= \frac{1}{x^2 + y^2} \left( - y z \, (ydy+zdz) - x^2 z dy +ydz\right) =ydz-zdy, \nn\\
* dy &=\frac{1}{x^2 + y^2} \left( y^2 z dx + xz (xdx+zdz) - x dz \right)=zdx-xdz. \label{stardx.dy.dz.1}
\end{align}
One may write~\eqref{stardx.dy.dz.1} in terms of
%%%%%%%%%%
\begin{align}
\boxed{\quad * dx_i = \varepsilon_{ijk} \, x_j \, dx_k\,, \quad }\label{stardx.dy.dz.3}
\end{align}
%%%%%%%%%%
where $x_1 = x$, $x_2 = y$, $x_3 = z$, and $\varepsilon_{123} = \varepsilon_{231} = \varepsilon_{312} = 1$. 

%%%%%%%%%%%%%%%%%%%%%%%%%%%%%%%%%%%%%%%%%%%%%%%%%%%%%%%%%%%%%%%%%
\subsection{Hodge dual via elimination of \texorpdfstring{$z$}{z}}
%%%%%%%%%%%%%%%%%%%%%%%%%%%%%%%%%%%%%%%%%%%%%%%%%%%%%%%%%%%%%%%%%

An alternative and more direct approach to compute the Hodge dual of differential forms on the unit sphere is to eliminate the coordinate $z$ using the constraint $x^2 + y^2 + z^2 = 1$, and hence $dz = -\frac{x \, dx + y \, dy}{z}$.  This allows us to express the induced metric solely in terms of $x$ and $y$:
\begin{align}
ds^2 &=  \left( 1 + \frac{x^2}{z^2} \right) dx^2 + 2 \frac{x y}{z^2} dx dy + \left( 1 + \frac{y^2}{z^2} \right) dy^2 \nn\\
&= \frac{1 - y^2}{z^2} dx^2 + 2 \frac{xy}{z^2} dx dy + \frac{1 - x^2}{z^2} dy^2.
\end{align}
Thus, the induced $2D$ metric can be written as
\begin{align}
ds^2 = h_{ab} dx^a dx^b, \qquad \text{with} \quad a, b = 1,2, \quad x^1 = x, \quad x^2 = y,
\end{align}
where the metric components are
\begin{align}
h_{11} = \frac{1 - y^2}{z^2}, \qquad
h_{12} = h_{21} = \frac{xy}{z^2}, \qquad
h_{22} = \frac{1 - x^2}{z^2}.
\end{align}

Using the definition~\eqref{dual}, the Hodge dual of a one-form can be written as
\begin{align}
* dx^1 = \epsilon_1{}^1 \, dx + \epsilon_2{}^1 \, dy, \qquad
* dx^2 = \epsilon_1{}^2 \, dx + \epsilon_2{}^2 \, dy. \label{*dx,*dy2}
\end{align}
To compute the components of the Levi-Civita tensor $\epsilon_a{}^b$, we use the relation $\epsilon_a{}^b = h_{ac} \, \epsilon^{cb}$ and $\epsilon^{ab} = \frac{1}{\sqrt{h}} \, \varepsilon^{ab}$, where $\varepsilon^{12} = -\varepsilon^{21} = 1$. This gives:
\begin{align}
\epsilon_1{}^1 &= h_{11} \epsilon^{11} + h_{12} \epsilon^{21} = h_{12} \epsilon^{21} = \frac{h_{12}}{\sqrt{h}} \varepsilon^{21} = -\frac{h_{12}}{\sqrt{h}}, \nn\\
\epsilon_1{}^2 &= h_{11} \epsilon^{12} + h_{12} \epsilon^{22} = h_{11} \epsilon^{12} = \frac{h_{11}}{\sqrt{h}} \varepsilon^{12} = \frac{h_{11}}{\sqrt{h}}, \nn\\
\epsilon_2{}^1 &= h_{22} \epsilon^{21} + h_{21} \epsilon^{11} = h_{22} \epsilon^{21} = \frac{h_{22}}{\sqrt{h}} \varepsilon^{21} = -\frac{h_{22}}{\sqrt{h}}, \nn\\
\epsilon_2{}^2 &= h_{21} \epsilon^{12} + h_{22} \epsilon^{22} = h_{21} \epsilon^{12} = \frac{h_{21}}{\sqrt{h}} \varepsilon^{12} = \frac{h_{21}}{\sqrt{h}}. \label{Lev3}
\end{align}
Substituting the expressions from~\eqref{*dx,*dy2} and~\eqref{Lev3}, we obtain the Hodge duals of $dx$ and $dy$ as:
\begin{align}
* dx = - \frac{h_{12}}{\sqrt{h}} \, dx - \frac{h_{22}}{\sqrt{h}} \, dy, \qquad
* dy = \frac{h_{11}}{\sqrt{h}} \, dx + \frac{h_{21}}{\sqrt{h}} \, dy. \label{*dx,*dy3}
\end{align}
The determinant of the induced metric is straightforward to compute:
%%%%%%%%
\begin{align}
\det(h) = h_{11}h_{22} - h_{12}^2= \frac{1}{z^4} \left[ 1 - x^2 - y^2 \right] = \frac{1}{z^2}
 \qquad \Rightarrow \qquad \sqrt{h} = \frac{1}{z}. \label{det.h}
\end{align}
%%%%%%%%
Substituting these results back into~\eqref{*dx,*dy3}, we find
\begin{align}
* dx = -\frac{xy}{z} \, dx - \frac{1 - x^2}{z} \, dy, \qquad
* dy = \frac{1 - y^2}{z} \, dx + \frac{xy}{z} \, dy, \label{stardx.dy.dz.2}
\end{align}
which matches the result obtained earlier by the more explicit computation in~\eqref{stardx.dy.dz}, using the identity $x^2 z^2 + y^2 = (x^2 + y^2)(1 - x^2)$.

%%%%%%%%%%%%%%%%%%%%%%%%%%%%%%%%%%%%%%%%%%%%%%%%%%%%%%%%%%%
\subsection{A diagonal metric in three dimensions}
%%%%%%%%%%%%%%%%%%%%%%%%%%%%%%%%%%%%%%%%%%%%%%%%%%%%%%%%%%%%

We now generalize the previous discussion by considering a diagonal (but non-Euclidean) metric in three dimensions:
\begin{align}
ds^2 = g_{ij} \, dx^i dx^j, \qquad \text{with} \quad g_{ij} = 0 \ \text{for } i \neq j, \qquad \sum_i (x^i)^2 = 1.
\end{align}
Substituting into the metric, after eliminating $dz$ yields:
%%%%%%%
\begin{align}
ds^2 &= g_{11} \, dx^2 + g_{22} \, dy^2 + g_{33} \left( \frac{x \, dx + y \, dy}{z} \right)^2 \nn\\
%%%%%%%%%
&= h_{11} \, dx^2 + 2 h_{12} \, dx \, dy + h_{22} \, dy^2,
\end{align}
%%%%%%%
where the induced $2D$ metric components are
%%%%%%%
\begin{align}
h_{11} = g_{11} + \frac{x^2}{z^2} g_{33}, \qquad
h_{12} = \frac{xy}{z^2} g_{33}, \qquad
h_{22} = g_{22} + \frac{y^2}{z^2} g_{33}.
\end{align}
The determinant of this induced metric is:
%%%%%%%%%%%%
\begin{align}
\det h = h_{11} h_{22} - h_{12}^2 
= \frac{1}{z^2} \left( x^2 g_{22} g_{33} + y^2 g_{11} g_{33} + z^2 g_{11} g_{22} \right).
\end{align}
%%%%%%%%%%%%
Let us compute $*dz$, which turns out to be the most symmetric, since it is purely written in terms of $dx$ and $dy$.  From the previously established relation~\eqref{*dx,*dy3}, we have:
%%%%%%%
\begin{align}
* dz &= -\frac{1}{z} \left( x \, *dx + y \, *dy \right) \nn\\
&= -\frac{1}{z} \left[ x \left( -\frac{h_{12}}{\sqrt{g}} dx - \frac{h_{22}}{\sqrt{g}} dy \right) + y \left( \frac{h_{11}}{\sqrt{g}} dx + \frac{h_{21}}{\sqrt{g}} dy \right) \right] \nn\\
&= \frac{1}{\sqrt{g} z} \left[ (x h_{12} - y h_{11}) \, dx + (x h_{22} - y h_{12}) \, dy \right].
\end{align}
%%%%%%%
Substituting the explicit forms of $h_{ab}$ yields:
%%%%%%%
\begin{align}
* dz &= \frac{1}{\sqrt{g} z} \left[ x \cdot \frac{xy}{z^2} g_{33} - y \left( g_{11} + \frac{x^2}{z^2} g_{33} \right) \right] dx \nn\\
&\quad + \frac{1}{\sqrt{g} z} \left[ x \left( g_{22} + \frac{y^2}{z^2} g_{33} \right) - y \cdot \frac{xy}{z^2} g_{33} \right] dy \nn\\
&= \frac{1}{\sqrt{g} z} \left( -y g_{11} \, dx + x g_{22} \, dy \right).
\end{align}
%%%%%%%
Thus, the Hodge dual of $dz$ becomes:
\begin{align}
\boxed{\quad 
*dz = \frac{-y g_{11} \, dx + x g_{22} \, dy}{\sqrt{x^2 g_{22} g_{33} + y^2 g_{11} g_{33} + z^2 g_{11} g_{22}}}
\quad} \,. \label{D=3explicit}
\end{align}
%%%%%%%
Analogous expressions for $*dx$ and $*dy$ can be derived similarly. As a consistency check, setting $g_{ij} = \delta_{ij}$ in~\eqref{D=3explicit} recovers the result of~\eqref{stardx.dy.dz.3}.

%%%%%%%%%%%%%%%%%%%%%%%%%%%%%%%%%%%%%%%%%%%%%%%%%%%%%%%%%%%%%%%%
\section{Finding the Hodge dual with a constraint}  \label{sec.non-cov}
%%%%%%%%%%%%%%%%%%%%%%%%%%%%%%%%%%%%%%%%%%%%%%%%%%%%%%%%%%%%%%%%

We now aim to derive the Hodge dual of a $p$-form using an explicit notation that emphasizes the inverse metric components, a convention that proves especially convenient in the presence of constraints. In this section, we adopt a convention in which the inverse metric is explicitly denoted by $g^{-1}$, rather than by raising indices. For example, instead of writing $A^i = g^{ij} A_j$, we write $g^{-1}_{ij} A_j$ to indicate the corresponding contravariant vector. 

To emphasize the use of the inverse metric $g^{-1}_{ij}$ explicitly and eliminate index-raising ambiguity, we rewrite the expression in a manifestly non-covariant form:
\begin{align}
* (dx^{i_1} \wedge \cdots \wedge dx^{i_p}) 
= \frac{1}{q!} \, g^{-1}_{i_1 k_1} \cdots g^{-1}_{i_p k_p} \, \sqrt{|g|} \, \varepsilon_{j_1 \cdots j_q k_1 \cdots k_p} \, dx^{j_1} \wedge \cdots \wedge dx^{j_q}.
\end{align}
%%%%%%
This formulation treats all indices with equal weight and avoids the ambiguity of index position (upstairs vs. downstairs), which is especially helpful when working in constrained coordinate systems or when interpreting the differential forms geometrically.

Our objective is to compute the Hodge dual of the following $p$-form:
\begin{align}
* (dx^{i_1} \wedge \cdots \wedge dx^{i_p}),
\end{align}
in the presence of a constraint and a general (not necessarily flat) background metric
\begin{align}
ds^2 = g_{ij} \, dx^i dx^j,
\end{align}
subject to the condition
%%%%%%%%%%%%%%%%
\begin{align}
x^i x^i = 1.
\end{align}
%%%%%%%%%%%%%%%%
The constraint implies
\begin{align}
x^i dx^i = 0,
\end{align}
which means that among the $D$ coordinate one-forms $dx^i$, only $D-1$ are linearly independent. Consequently, the exterior algebra is effectively restricted to a $(D-1)$-dimensional hypersurface embedded in $\mathbb{R}^D$, and the Hodge dual of a $p$-form on this hypersurface becomes a $(q-1)$-form, with
\begin{align}
D = p + q.
\end{align}

Furthermore, the usual identity for the double Hodge dual operation still applies in this restricted setting:
\begin{align}
** \omega_{(p)} = (-1)^{p(D-1 - p) + t} \, \omega_{(p)},
\end{align}
where $t$ denotes the number of timelike directions in the ambient metric.

With these two observations in place, we proceed to construct an explicit ansatz for the Hodge dual on the constrained hypersurface.

We now propose the following ansatz for the Hodge dual of a $p$-form in the presence of the constraint $x^i x^i = 1$:
\begin{align}
* (dx^{i_1} \wedge \cdots \wedge dx^{i_p}) = \frac{1}{(q-1)!} \, \sqrt{|g|} \, \varepsilon_{j j_1 \cdots j_{q-1} k_1 \cdots k_p}
\, V_j \, g^{-1}_{i_1 k_1} \cdots g^{-1}_{i_p k_p} \,
dx^{j_1} \wedge \cdots \wedge dx^{j_{q-1}}. \label{ansatz}
\end{align}
Here, $|g| = (-1)^t \det g$ accounts for the signature of the metric, and $V_j$ is a vector field introduced to contract the extra free index from the Levi-Civita symbol, reflecting the fact that we are working in a $(D-1)$-dimensional constrained subspace. The ansatz above satisfies two essential properties:
\begin{enumerate}
    \item It is totally antisymmetric in the indices $(i_1 \cdots i_p)$, as guaranteed by the structure of the Levi-Civita symbol.
    \item It yields a $(q-1)$-form, as appropriate for Hodge duality on the constrained $(D-1)$-dimensional manifold.
\end{enumerate}

\vspace{0.5em}
\noindent
To determine the explicit form of $V_j$, we apply the Hodge dual operation twice. Starting from our ansatz, we compute
\begin{align}
* (dx^{i_1} \wedge \cdots \wedge dx^{i_p}) 
=& \frac{1}{(q-1)!} \sqrt{|g|} \, \varepsilon_{j j_1 \cdots j_{q-1} k_1 \cdots k_p} \,
V_j \, g^{-1}_{i_1 k_1} \cdots g^{-1}_{i_p k_p} \,
* (dx^{j_1} \wedge \cdots \wedge dx^{j_{q-1}})
\nn\\
=& \frac{1}{(q-1)!} \sqrt{|g|} \, \varepsilon_{j j_1 \cdots j_{q-1} k_1 \cdots k_p} \,
V_j \, g^{-1}_{i_1 k_1} \cdots g^{-1}_{i_p k_p} \,
\nn\\
&\quad \times \frac{1}{p!} \sqrt{|g|} \, \varepsilon_{r r_1 \cdots r_p s_1 \cdots s_{q-1}} \,
V_r \, g^{-1}_{j_1 s_1} \cdots g^{-1}_{j_{q-1} s_{q-1}} \,
dx^{r_1} \wedge \cdots \wedge dx^{r_p}. \label{**1}
\end{align}
This sets the stage for determining the precise form of $V_j$ by demanding consistency with the known identity for the double Hodge dual on a $(D-1)$-dimensional manifold. We will now analyze this expression further to constrain $V_j$.

Note that we now have two sets of indices: $(i,j,k)$ and $(j,r,s)$. Thus, the expression in \eqref{**1} becomes
\begin{align}
* * (dx^{i_1} \wedge \cdots \wedge dx^{i_p}) 
=& \frac{1}{p!} \frac{1}{(q-1)!} (-1)^t \det g \, (-1)^{p(q-1)}
\nn\\
&\times \varepsilon_{j j_1 \cdots j_{q-1} k_1 \cdots k_p} \varepsilon_{r s_1 \cdots s_{q-1} r_1 \cdots r_p} V_j V_r
\nn\\
&\times g^{-1}_{j_1 s_1} \cdots g^{-1}_{j_{q-1} s_{q-1}} \,
g^{-1}_{k_1 i_1} \cdots g^{-1}_{k_p i_p} \,
dx^{r_1} \wedge \cdots \wedge dx^{r_p}. \label{**2}
\end{align}
%%%%%%%%%%
On the other hand, the standard identity for the double Hodge dual on a $(D-1)$-dimensional manifold with $t$ time-like directions is
%%%%%%%%%%
\begin{align}
* * (dx^{i_1} \wedge \cdots \wedge dx^{i_p}) 
=& (-1)^{p(q-1) + t} dx^{i_1} \wedge \cdots \wedge dx^{i_p} \nn\\
%%%%
=& (-1)^{p(q-1)+t} \delta_{i_1 \cdots i_p}^{r_1 \cdots r_p} dx^{r_1} \wedge \cdots \wedge dx^{r_p}. \label{**3}
\end{align}
%%%%%%%%%%
Comparing \eqref{**2} and \eqref{**3}, we require that the prefactors multiplying $dx^{r_1} \wedge \cdots \wedge dx^{r_p}$ match. This yields the identity:
%%%%%%%%%%
\begin{align}
\frac{1}{p!} \frac{1}{(q-1)!} & \det g \,
\varepsilon_{j j_1 \cdots j_{q-1} k_1 \cdots k_p} \,
\varepsilon_{r s_1 \cdots s_{q-1} r_1 \cdots r_p} \,
V_j V_r \,
g^{-1}_{j_1 s_1} \cdots g^{-1}_{j_{q-1} s_{q-1}} \,
g^{-1}_{k_1 i_1} \cdots g^{-1}_{k_p i_p}
= \delta_{i_1 \cdots i_p}^{r_1 \cdots r_p}. \label{**4}
\end{align}
%%%%%%%%%%
To further constrain the structure, we contract both sides of \eqref{**4} with $\delta^{i_1 \cdots i_p}_{r_1 \cdots r_p}$, yielding
\begin{align}
\frac{1}{p!} \frac{1}{(q-1)!} & \det g \,
\varepsilon_{j j_1 \cdots j_{q-1} k_1 \cdots k_p} \,
\varepsilon_{r s_1 \cdots s_{q-1} i_1 \cdots i_p} \,
V_j V_r \,
g^{-1}_{j_1 s_1} \cdots g^{-1}_{j_{q-1} s_{q-1}} \,
g^{-1}_{k_1 i_1} \cdots g^{-1}_{k_p i_p} \nn\\
%%%%
&= \binom{D-1}{p}. \label{**5}
\end{align}
%%%%%%%%%%
There is an important caveat to observe here. If there is no restriction on the set $\{i_1, \cdots, i_p\}$, then we have the identity
\begin{align}
\delta^{i_1 \cdots i_p}_{i_1 \cdots i_p} = \binom{D}{p},
\end{align}
as one would be choosing $p$ indices out of $D$ total possibilities. However, due to the presence of the vector index $r$ in the ansatz, the indices $\{i_1, \cdots, i_p\}$ are constrained to lie within the subspace orthogonal to the constraint direction. That is, the index $r$ already occupies one of the $D$ slots, effectively reducing the available directions to $D-1$. Consequently, the correct counting becomes
\begin{align}
\delta^{i_1 \cdots i_p}_{i_1 \cdots i_p} = \binom{D-1}{p}.
\end{align}
%%%%%%%%
The first line of \eqref{**5} can now be simplified using the following identity:
\begin{align}
g_{jr}= \frac{1}{\det g^{-1} (D-1)!}\,\varepsilon_{j j_1 \cdots j_{q-1} k_1 \cdots k_p} \varepsilon_{r s_1 \cdots s_{q-1} i_1 \cdots i_p} 
g^{-1}_{j_1 s_1} \cdots g^{-1}_{j_{q-1} s_{q-1}} g^{-1}_{k_1 i_1} \cdots g^{-1}_{k_p i_p}. \label{**6}
\end{align}
%%%%%%%%
Inserting the result from \eqref{**6} into \eqref{**5}, we arrive at the identity:
\begin{align}
\frac{1}{p!} \frac{1}{(q-1)!}  \det g \, \det g^{-1} (D-1)! \, g_{jr} V_j V_r =
\binom{D-1}{p}. \label{**7}
\end{align}
%%%%%%%%
This leads to a remarkably simple constraint on the vector $V_j$:
\begin{align}
\boxed{g_{ij} V_i V_j = 1}\,. \label{**8}
\end{align}
This result provides the necessary condition that ensures the consistency of our ansatz for the Hodge dual in the presence of a constraint such as $x^i x^i = 1$.

%%%%%%%%%%%%%%%%%%%%%%%%%%%%%%%%%%%%%%%%%%%%%%%%%%
\subsection{\texorpdfstring{Finding an ansatz for $V_i$}{}}
%%%%%%%%%%%%%%%%%%%%%%%%%%%%%%%%%%%%%%%%%%%%%%%%%%

A natural first ansatz is to take
\begin{align}
V^i = \Delta \, x^i,
\end{align}
where the normalization factor $\Delta$ is determined by the constraint
\begin{align}
g_{ij} V^i V^j = \Delta^2 \, g_{ij} x^i x^j = 1 
\quad \Rightarrow \quad 
\Delta = \frac{1}{\sqrt{g_{ij} x^i x^j}}\,.
\end{align}
At first glance, this choice appears reasonable and consistent with the normalization condition \eqref{**8}. However, there is a crucial subtlety to consider.

Throughout this section, we have not explicitly used the constraint $x^i x^i = 1$, except in determining the form degree of the Hodge dual, which becomes a $(D-1-p)$-form rather than a $(D-p)$-form. Nevertheless, the constraint has further implications.

A direct consequence of the identity $x^i x^i = 1$ is the differential relation $x^i dx^i = 0$. Namely,
\begin{align}
x^\ell * (dx^{i_1} \wedge \cdots \wedge dx^{i_p}) = 0,
\end{align}
whenever $\ell \in \{i_1, \cdots, i_p\}$, say $\ell = i_s$. Let us check whether this condition is satisfied by our ansatz \eqref{ansatz} with the choice of $V^i = \Delta \, x^i$. Substituting into the expression yields:
\begin{align}
x^{i_s} * (dx^{i_1} \wedge \cdots \wedge dx^{i_p}) 
&\propto x^{i_s} \varepsilon_{j j_1 \cdots j_{q-1} k_1 \cdots k_p}
\, V_j \, g^{-1}_{i_1 k_1} \cdots g^{-1}_{i_p k_p}
dx^{j_1} \wedge \cdots \wedge dx^{j_{q-1}}. \label{vanish}
\end{align}
Now, if we take $V^i = \Delta x^i$, then the expression becomes
\begin{align}
\propto \varepsilon_{j j_1 \cdots j_{q-1} k_1 \cdots k_p}
\, x^{i_s} x^j \, g^{-1}_{i_1 k_1} \cdots g^{-1}_{i_p k_p}
dx^{j_1} \wedge \cdots \wedge dx^{j_{q-1}}.
\end{align}
This clearly does not vanish automatically, due to lack of any anti-symmetry between between $x^{i_s}$ and $x^j$. Thus, the ansatz $V^i = \Delta x^i$ fails to satisfy this important orthogonality condition imposed by the constraint $x^i x^i = 1$.

This failure indicates that a more refined choice for $V^i$ is required—one that not only satisfies the normalization condition $g_{ij} V^i V^j = 1$, but also ensures that the contraction of $x^\ell$ with the Hodge dual vanishes when $\ell$ appears among the original form indices.

However, the resolution is straightforward. Let us assume  
\begin{align}
V^i = \kappa \, g^{-1}_{ij} x^j\,.
\end{align}
It is now easy to see that the right-hand side of \eqref{vanish} vanishes, as it becomes proportional to
\begin{align}
&\varepsilon_{j j_1 \cdots j_{q-1} k_1 \cdots k_p}
\, g^{-1}_{j k} x^k x^{i_s} \, g^{-1}_{i_1 k_1} \cdots g^{-1}_{i_s k_s} \cdots g^{-1}_{i_p k_p}
dx^{j_1} \wedge \cdots \wedge dx^{j_{q-1}} \nn\\
%%%%%
&= \varepsilon_{j j_1 \cdots j_{q-1} k_1 \cdots k_p}
\, \underbrace{g^{-1}_{j k} x^k\, g^{-1}_{k_s i_s} x^{i_s}}_{\text{symmetric under } j \leftrightarrow k_s}
\, g^{-1}_{i_1 k_1} \cdots g^{-1}_{i_p k_p}
dx^{j_1} \wedge \cdots \wedge dx^{j_{q-1}} = 0,
\end{align}
since the expression inside the contraction is symmetric under the exchange $j \leftrightarrow k_s$, while the Levi-Civita symbol is antisymmetric, hence the full term vanishes.

Finding $\kappa$ is straightforward from the normalization condition $g_{ij} V^i V^j = 1$:
\begin{align}
g_{ij} \kappa g^{-1}_{im} x^m \, \kappa g^{-1}_{jn} x^n 
= \kappa^2 \delta_{jm} g^{-1}_{jn} x^m x^n 
= \kappa^2 g^{-1}_{mn} x^m x^n = 1 
\quad \Rightarrow \quad 
\kappa = \frac{1}{\sqrt{g^{-1}_{mn} x^m x^n}}\,.
\end{align}
Therefore,
\begin{align}
V^i = \frac{1}{\sqrt{g^{-1}_{mn} x^m x^n}} \, g^{-1}_{ij} x^j\,.
\end{align}
Finally, the Hodge dual becomes:
\begin{align}
\boxed{
*(dx^{i_1} \wedge \cdots \wedge dx^{i_p}) 
= \frac{1}{(q-1)!} \sqrt{\frac{(-1)^t \det g}{g^{-1}_{mn} x^m x^n}} 
\, \varepsilon_{j j_1 \cdots j_{q-1} k_1 \cdots k_p} 
g^{-1}_{j k} x^k\, g^{-1}_{i_1 k_1} \cdots g^{-1}_{i_p k_p}
dx^{j_1} \wedge \cdots \wedge dx^{j_{q-1}} 
}\,. \label{ansatz.final}
\end{align}
%%%%%%%%%%%%%%%%%%%%%%%%%%%%%%%%%%%%%%%%%%%%%%%%%%%%%%
\subsection{Some special results}
%%%%%%%%%%%%%%%%%%%%%%%%%%%%%%%%%%%%%%%%%%%%%%%%%%%%%%

Let us consider the special case of $p=1$. We obtain:
\begin{align}
\boxed{
*dx^{i}  = \frac{1}{(D-2)!} \sqrt{\frac{(-1)^t \det g}{g^{-1}_{mn} x^m x^n}}\, g^{-1}_{i k} g^{-1}_{j \ell} x^\ell\, \varepsilon_{j j_1 \cdots j_{D-2} k} \,
dx^{j_1} \wedge \cdots \wedge dx^{j_{D-2}}
}\,. \label{p=1}
\end{align}
%%%%%%%
Now, let us specialize to the case $D=3$:
\begin{align}
*dx^{i} = \sqrt{\frac{(-1)^t \det g}{g^{-1}_{mn} x^m x^n}}\, g^{-1}_{i k} g^{-1}_{j \ell} x^\ell\, \varepsilon_{j r k} \, dx^{r}
\,. \label{D=3,p=1}
\end{align}
%%%%%%%%%%%%%%%%%%%
In the further special case where the metric is diagonal, i.e., $g_{ij} = 0$ for $i \neq j$, we have:
\begin{align}
g^{-1}_{11} = \frac{1}{g_{11}}, \quad 
g^{-1}_{22} = \frac{1}{g_{22}}, \quad 
g^{-1}_{33} = \frac{1}{g_{33}}, \quad 
\det g = g_{11} g_{22} g_{33},
\end{align}
and
\begin{align}
\sqrt{\frac{\det g}{g^{-1}_{mn} x^m x^n}} 
= \frac{g_{11} g_{22} g_{33}}{\sqrt{x^2 g_{22} g_{33} + y^2 g_{11} g_{33} + z^2 g_{11} g_{22}}}.
\end{align}

Assuming a Euclidean signature ($t=0$), the explicit expressions for the Hodge duals become:
\begin{align}
*dx &= *dx^1 = \frac{y g_{33} dz - z g_{22} dy}{\sqrt{x^2 g_{22} g_{33} + y^2 g_{11} g_{33} + z^2 g_{11} g_{22}}}, \nn\\
*dy &= *dx^2 = \frac{z g_{11} dx - x g_{33} dz}{\sqrt{x^2 g_{22} g_{33} + y^2 g_{11} g_{33} + z^2 g_{11} g_{22}}}, \nn\\
*dz &= *dx^3 = \frac{x g_{22} dy - y g_{11} dx}{\sqrt{x^2 g_{22} g_{33} + y^2 g_{11} g_{33} + z^2 g_{11} g_{22}}}.
\end{align}
These expressions exactly match those derived previously in equation~\eqref{D=3explicit}.

%%%%%%%%%%%%%%%%%%%%%%%%%%%%%%%%%%%%%%%%%%%%%%%%%%
\subsection{Finding the volume form}
%%%%%%%%%%%%%%%%%%%%%%%%%%%%%%%%%%%%%%%%%%%%%%%%%%

Consider the special case of $p = 0$, for which the Hodge dual of a scalar (zero-form) is the volume form. From our general ansatz, we obtain
\begin{align}
*\oneone 
=& \frac{1}{(D-1)!} \sqrt{\frac{(-1)^t \det g}{g^{-1}_{mn} x^m x^n}} 
\, \varepsilon_{j j_1 \cdots j_{D-1}} 
g^{-1}_{j k} x^k\, 
dx^{j_1} \wedge \cdots \wedge dx^{j_{D-1}}. \label{volform1}
\end{align}
%%%%
Multiplying both sides by $x_\ell x_\ell = 1$ allows us to exploit the identity
\begin{align}
\varepsilon_{j j_1 \cdots j_{D-1}} x_\ell 
= \sum_{r=1}^{D-1} \varepsilon_{j \cdots \ell \cdots j_{D-1}} x_{j_r}
+ \varepsilon_{\ell j_1 \cdots j_{D-1}} x_j,
\end{align}
which follows from the Schouten identity:
\begin{align}
\varepsilon_{[\mu_1 \cdots \mu_D} A_{\mu]} = 0
\quad \Rightarrow \quad
\varepsilon_{\mu_1 \cdots \mu_D} A_{\mu} 
= \sum_{r=1}^{D} \varepsilon_{\mu_1 \cdots \mu_{r-1} \mu \mu_{r+1} \cdots \mu_D} A_{\mu_r}.
\end{align}
Hence we have
%%%%%%%%%%%
\begin{align}
*\oneone 
=& \frac{1}{(D-1)!} \sqrt{\frac{(-1)^t \det g}{g^{-1}_{mn} x^m x^n}} 
\, \Big(\sum_{r=1}^{D-1} \varepsilon_{j \cdots \ell \cdots j_{D-1}} x_{j_r}
+ \varepsilon_{\ell j_1 \cdots j_{D-1}} x_j \Big)
g^{-1}_{j k} x^k\, x_\ell\,
dx^{j_1} \wedge \cdots \wedge dx^{j_{D-1}}. \nn
\end{align}
%%%%
Because $x^{j_r} dx^{j_r} = 0$ due to the constraint $x^i dx^i = 0$, all terms involving $x_{j_r} dx^{j_r}$ vanish, leaving only the last term in the parenthesis above. Therefore, the volume form becomes
\begin{align}
*\oneone
=& \frac{1}{(D-1)!} \sqrt{\frac{(-1)^t \det g}{g^{-1}_{mn} x^m x^n}} 
\, \varepsilon_{\ell j_1 \cdots j_{D-1}}\, \, x^\ell\,x^j g^{-1}_{j k} x^k
dx^{j_1} \wedge \cdots \wedge dx^{j_{D-1}}. \label{volform2}
\end{align}
Hence, the volume form on the unit sphere embedded in a $D$-dimensional space becomes
\begin{align}
*\oneone
= \frac{1}{(D-1)!} \sqrt{(-1)^t \det g \, g^{-1}_{mn} x^m x^n} 
\, \varepsilon_{i i_1 \cdots i_{D-1}}\, x^i 
dx^{i_1} \wedge \cdots \wedge dx^{i_{D-1}}.
\end{align}

Alternatively, in a fully covariant language using the Levi-Civita tensor and explicit metric notation, we can write
\begin{align}
\boxed{
*\oneone 
= \frac{1}{(D-1)!} \sqrt{g^{\rho \sigma} x^\rho x^\sigma} 
\, \epsilon_{\mu \mu_1 \cdots \mu_{D-1}}\, x^\mu 
dx^{\mu_1} \wedge \cdots \wedge dx^{\mu_{D-1}}.
}
\end{align}
This gives a compact and manifestly covariant expression for the induced volume form on the constrained hypersurface $x^i x^i = 1$.

%%%%%%%%%%%%%%%%%%%%%%%%%%%%%%%%%%%%%%%%%%%%%%%%%%%%%%%%%%%%%%%%%%%%
\section{Conclusion} \label{con}
%%%%%%%%%%%%%%%%%%%%%%%%%%%%%%%%%%%%%%%%%%%%%%%%%%%%%%%%%%%%%%%%%%%%

In this work, we have derived a general expression for the Hodge dual of an arbitrary $p$-form in a $D$-dimensional manifold, subject to the constraint $x^i x^i = 1$, as naturally arises in spherical compactifications. Starting from a systematic review of conventional Hodge duality, we built the necessary machinery to extend the formalism to settings in which the differential forms are not fully independent due to the presence of embedding constraints.

By constructing a consistent ansatz for the constrained Hodge dual and demanding compatibility with the fundamental identity $\ast \ast = \pm 1$ on $p$-forms, we identified a normalized radial vector field $V^i$ that governs the projection from the ambient space to the constrained submanifold. This geometric insight allowed us to reinterpret the Hodge dual operation in the presence of a constraint as a projection of the ambient dual onto the tangential directions of the sphere.

Our formalism correctly reproduces known results in low-dimensional cases, such as three-dimensional Euclidean space, and generalizes naturally to arbitrary diagonal or non-diagonal metrics. The final expression in equation~\eqref{ansatz.final} provides a compact and universal formula that captures the essential features of constrained Hodge duality and can be applied in a wide range of contexts.

Beyond its mathematical utility, this result has significant implications for compactified theories in high-energy physics, particularly in supergravity, string theory, and M-theory. The ability to explicitly compute Hodge duals in constrained geometries facilitates the construction of consistent truncations, the formulation of flux quantization conditions, and the derivation of effective lower-dimensional actions. It may also shed light on the structure of topological terms, duality relations, and Chern-Simons couplings in reduced theories.

Moreover, our derivation clarifies longstanding technical issues encountered in consistent Kaluza-Klein reductions on spheres and offers a general formalism that complements and extends the results found in prior literature such as~\cite{Cvetic:1999au,Cvetic:2000ah,Cvetic:2000tb,Cvetic:2003jy,Gibbons:2003gp,Azizi-Pope}. The connection between the radial vector $V^i$ and the constrained Hodge dual provides a geometric mechanism that may prove beneficial in applications involving constrained brane embeddings, flux stabilization, and the study of effective potential landscapes.

Several directions for future research are worth pursuing. These include generalizing the construction to more intricate embedding constraints, such as those found in squashed spheres or non-spherical coset spaces, as well as to curved internal manifolds with torsion or warping. Another promising direction is the extension to generalized geometry and double field theory, where analogous constraints appear in the context of dual coordinates. Finally, applications of this formalism to nontrivial compactifications involving fluxes, topological couplings, and moduli stabilization merit detailed investigation.

%%%%%%%%%%%%%%%%%%%%%%%%%%%%%%%%%%%%%%%%%%%%%%%%%%%%%%
%%%%%%%%%%%%%%%%%%%%%%%%%%%%%%%%%%%%%%%%%%%%%%%%%%%%
\section*{Acknowledgments}
%%%%%%%%%%%%%%%%%%%%%%%%%%%%%%%%%%%%%%%%%%%%%%%%%%%%%%
%%%%%%%%%%%%%%%%%%%%%%%%%%%%%%%%%%%%%%%%%%%%%%%%%%%%

I am grateful to Chris Pope for illuminating discussions and insights. This work was supported by the Robert A. Welch Foundation (Grant No. A-1261) and the National Science Foundation (Grant No. PHY-2013771).

%%%%%%%%%%%%%%%%%%%%%%%%%%%%%%%%%%%%%%%%%%%%%%%%%%%%%%
%%%%%%%%%%%%%%%%%%%%%%%%%%%%%%%%%%%%%%%%%%%%%%%%%%%%
\appendix 
%%%%%%%%%%%%%%%%%%%%%%%%%%%%%%%%%%%%%%%%%%%%%%%%%%%%%%

%%%%%%%%%%%%%%%%%%%%%%%%%%%%%%%%%%%%%%%%%%%%%%%%%%%%%%%%  
\section{Some useful theorems}   \label{app}
%%%%%%%%%%%%%%%%%%%%%%%%%%%%%%%%%%%%%%%%%%%%%%%%%%%%%%%%  

\noindent\textbf{Generalized Kronecker delta.}  
The generalized Kronecker delta is an extension of the usual Kronecker delta $\delta_i^j$, and is defined by  
%%%%  
\begin{align}  
\delta^{j_1 \cdots j_D}_{i_1 \cdots i_D} 
= \delta^{[j_1}_{[i_1} \cdots \delta^{j_D]}_{i_D]}  
= \delta^{[j_1}_{i_1} \cdots \delta^{j_D]}_{i_D},  
\end{align}  
%%%%  
where square brackets denote antisymmetrization with unit weight, i.e.,  
\begin{align}  
A_{[ij]} = \frac{1}{2}(A_{ij} - A_{ji}).  
\end{align}  

%%%%%%%%%%%%%%  
\noindent\textbf{Lemma.}  
Let $A_{i_1 \cdots i_D} = A_{[i_1 \cdots i_D]}$ be a totally antisymmetric tensor. Then,  
%%%%  
\begin{align}  
\delta^{j_1 \cdots j_D}_{i_1 \cdots i_D} A^{i_1 \cdots i_D} = A^{j_1 \cdots j_D}.  
\end{align}  
%%%%  
This identity follows directly from the definition of the generalized Kronecker delta and the antisymmetry of $A$.

%%%%%%%%%%%%%%%%%%%%%%%%%%%%%%%%%
\begin{theorem}
%%%%%%%%%%%%%%%%
\begin{align}
\delta_{i_1 \cdots i_p}^{i_1 \cdots i_p} = \binom{D}{p}
\end{align}
%%%%%%%%%%%%%%%%
\end{theorem}
%%%%%%%%%
\begin{proof}
The key observation is that for a fixed set of indices,
\begin{align}
\delta_{i_1 \cdots i_p}^{i_1 \cdots i_p} = 1,
\end{align}
where \emph{no summation} is implied. The number of distinct antisymmetric combinations of $p$ indices from a set of $D$ is simply the number of $p$-element subsets of a $D$-element set, which is $\binom{D}{p}$. This identity generalizes the well-known result $\delta_{ii} = D$. In particular, when $p = D$, one finds
\begin{align}
\delta_{i_1 \cdots i_D}^{i_1 \cdots i_D} = 1.
\end{align}
\end{proof}
%%%%%%%%%%%%%%
\begin{theorem}
%%%%%%%%%%%%%%%%
\begin{align}
\varepsilon_{i_1 \cdots i_D} \varepsilon_{j_1 \cdots j_D} = D! \, \delta^{j_1 \cdots j_D}_{i_1 \cdots i_D}
\end{align}
%%%%%%%%%%%%%%%%
\end{theorem}
%%%%%%%%%
\begin{proof}
Both sides are totally antisymmetric under permutations of the indices $i_1, \ldots, i_D$ and $j_1, \ldots, j_D$, so they must be proportional. To determine the proportionality constant, contract both sides with $\delta_{j_1 \cdots j_D}^{i_1 \cdots i_D}$:
%%%%%%%%%%%%%%%%
\begin{align*}
\text{RHS} &= \delta_{j_1 \cdots j_D}^{i_1 \cdots i_D} \delta^{j_1 \cdots j_D}_{i_1 \cdots i_D}
= \delta^{i_1 \cdots i_D}_{i_1 \cdots i_D} = \binom{D}{D} = 1, \\\\
\text{LHS} &= \varepsilon_{i_1 \cdots i_D} \varepsilon_{j_1 \cdots j_D} \delta_{j_1 \cdots j_D}^{i_1 \cdots i_D}
= \varepsilon_{i_1 \cdots i_D} \varepsilon^{i_1 \cdots i_D} = D!\,,
\end{align*}
%%%%%%%%%%%%%%%%
where we used the identity $\varepsilon_{i_1 \cdots i_D} \varepsilon^{i_1 \cdots i_D} = D!$\,. Thus, the proportionality constant is $D!$, which completes the proof.
\end{proof}

%%%%%%%%%%%%%%
\begin{theorem}
%%%%%%%%%%%%%%%%
\begin{align}
\boxed{\quad
\delta^{i_1 \cdots i_p k_1 \cdots k_q}_{i_1 \cdots i_p j_1 \cdots j_q}  = \frac{p! \, q!}{D!} \, \delta^{k_1 \cdots k_q}_{j_1 \cdots j_q}
\quad}
\end{align}
%%%%%%%%%%%%%%%%
\end{theorem}
\begin{proof}
Left-hand side is totally antisymmetric in the sets $(j_1, \ldots, j_q)$ and $(k_1, \ldots, k_q)$, so the right-hand side must be proportional to $\delta^{k_1 \cdots k_q}_{j_1 \cdots j_q}$. To determine the proportionality constant, contract both sides with $\delta^{j_1 \cdots j_q}_{k_1 \cdots k_q}$.

\noindent\textbf{Left-hand side:}
%%%%%%%%%%%%%%%%
\begin{align}
\delta^{i_1 \cdots i_p k_1 \cdots k_q}_{i_1 \cdots i_p j_1 \cdots j_q} \, \delta^{j_1 \cdots j_q}_{k_1 \cdots k_q}
&= \delta^{i_1 \cdots i_p j_1 \cdots j_q}_{i_1 \cdots i_p j_1 \cdots j_q} = \binom{D}{D} = 1.
\end{align}
%%%%%%%%%%%%%%%%
\noindent\textbf{Right-hand side:}
%%%%%%%%%%%%%%%%
\begin{align}
\left( \frac{p! \, q!}{D!} \right) \delta^{k_1 \cdots k_q}_{j_1 \cdots j_q} \, \delta^{j_1 \cdots j_q}_{k_1 \cdots k_q}
&= \left( \frac{p! \, q!}{D!} \right) \delta^{j_1 \cdots j_q}_{j_1 \cdots j_q} = \left( \frac{p! \, q!}{D!} \right) \binom{D}{q}.
\end{align}
%%%%%%%%%%%%%%%%
Using the identity $\binom{D}{q} = \frac{D!}{p! \, q!}$, the right-hand side simplifies to 1, matching the left-hand side. This confirms that the proportionality constant is indeed $\frac{p! \, q!}{D!}$, completing the proof.
\end{proof}

%%%%%%%%%%%%%%
\begin{theorem} 
Let $D = p + q$. Then,
\begin{align}
\boxed{\quad 
\varepsilon_{i_1 \cdots i_p j_1 \cdots j_q} \, \varepsilon^{i_1 \cdots i_p k_1 \cdots k_q} = p! \, q! \, \delta^{k_1 \cdots k_q}_{j_1 \cdots j_q}\,.
\quad }
\end{align}
\end{theorem}
%%%%%%%%%%%%%%
\begin{proof}
We observe that:
\begin{align}
\varepsilon_{i_1 \cdots i_p j_1 \cdots j_q} \, \varepsilon^{i_1 \cdots i_p k_1 \cdots k_q}
= D! \, \delta^{i_1 \cdots i_p k_1 \cdots k_q}_{i_1 \cdots i_p j_1 \cdots j_q},
\end{align}
where it follows from Theorem 3. Applying Theorem 4 to the generalized Kronecker delta gives:
\begin{align}
\delta^{i_1 \cdots i_p k_1 \cdots k_q}_{i_1 \cdots i_p j_1 \cdots j_q}
= \frac{p! \, q!}{D!} \, \delta^{k_1 \cdots k_q}_{j_1 \cdots j_q}.
\end{align}
Substituting this into the previous expression yields:
\begin{align}
\varepsilon_{i_1 \cdots i_p j_1 \cdots j_q} \, \varepsilon^{i_1 \cdots i_p k_1 \cdots k_q}
= p! \, q! \, \delta^{k_1 \cdots k_q}_{j_1 \cdots j_q},
\end{align}
which completes the proof. Note that Theorem 3 is a special case of Theorem 5 with $p = 0$ and $q = D$.
\end{proof}

%%%%%%%%%%%%%%
\begin{theorem} 
Let $D = p + q$. Then,
%%%%%%%%%
\begin{align}
\boxed{\quad 
\varepsilon_{i_1 \cdots i_p j_1 \cdots j_q} \, \varepsilon_{m_1 \cdots m_p n_1 \cdots n_q}
\, A_{i_1 m_1} \cdots A_{i_p m_p} A_{j_1 k_1} \cdots A_{j_q k_q}
= p! \, q! \, \det A \, \delta_{n_1 \cdots n_q}^{k_1 \cdots k_q}.
\quad}
\label{det.A.delta}
\end{align}
%%%%%%%%%
\end{theorem}
%%%%%%%%%%%
\begin{proof}
The left-hand side must be proportional to $ \delta_{n_1 \cdots n_q}^{k_1 \cdots k_q}$, since both expressions are antisymmetric in the indices $\{k_1, \ldots, k_q\}$ and $\{n_1, \ldots, n_q\}$.

To see this, assume for contradiction that a term survives where $k_r \in \{m_1, \ldots, m_p\}$—i.e., some $k_r = m_s$. Then the combination $A_{i_s m_s} A_{j_r k_r}$ becomes symmetric under the exchange $i_s \leftrightarrow j_r$, while the Levi-Civita symbol $\varepsilon_{i_1 \cdots i_p j_1 \cdots j_q}$ is antisymmetric under that same exchange. Hence the term must vanish, and only terms where $\{k_1, \ldots, k_q\}$ matches $\{n_1, \ldots, n_q\}$ survive, up to antisymmetric reshuffling. Therefore, the structure must be:
%%%%%%%%%
\begin{align}
\varepsilon_{i_1 \cdots i_p j_1 \cdots j_q} \, \varepsilon_{m_1 \cdots m_p n_1 \cdots n_q}
\, A_{i_1 m_1} \cdots A_{i_p m_p} A_{j_1 k_1} \cdots A_{j_q k_q}
= \kappa \, \delta_{n_1 \cdots n_q}^{k_1 \cdots k_q}.
\end{align}
%%%%%%%%%
To determine the proportionality constant $\kappa$, contract both sides with $\delta_{k_1 \cdots k_q}^{n_1 \cdots n_q}$. The right-hand side yields:
\begin{align}
\kappa \, \delta_{n_1 \cdots n_q}^{k_1 \cdots k_q} \delta_{k_1 \cdots k_q}^{n_1 \cdots n_q}
= \kappa \, \delta_{n_1 \cdots n_q}^{n_1 \cdots n_q}
= \kappa \, \binom{D}{q}.
\end{align}
%%%%%%%%%
Now consider the left-hand side:
\begin{align}
&\varepsilon_{i_1 \cdots i_p j_1 \cdots j_q} \, \varepsilon_{m_1 \cdots m_p k_1 \cdots k_q}
A_{i_1 m_1} \cdots A_{i_p m_p} A_{j_1 k_1} \cdots A_{j_q k_q} \delta_{k_1 \cdots k_q}^{n_1 \cdots n_q} \nn\\
&= \varepsilon_{m_1 \cdots m_p n_1 \cdots n_q} \, \varepsilon_{r_1 \cdots r_p s_1 \cdots s_q}
A_{m_1 r_1} \cdots A_{m_p r_p} A_{n_1 s_1} \cdots A_{n_q s_q}
= D! \, \det A.
\end{align}
%%%%%%%%%
Equating both sides:
\begin{align}
D! \, \det A = \kappa \, \binom{D}{q} = \kappa \, \frac{D!}{p! \, q!} \quad \Rightarrow \quad \kappa = p! \, q! \, \det A.
\end{align}
This completes the proof.
\end{proof}

%%%%%%%%%%%%%%
\begin{theorem} 
Let $D = p + q$, and let $t$ denote the number of time-like directions in the metric. Then, for a $p$-form expressed in local coordinates:
%%%%%%%%%
\begin{align}
* * (dx^{i_1} \wedge \cdots \wedge dx^{i_p}) = (-1)^{pq + t} \, dx^{i_1} \wedge \cdots \wedge dx^{i_p}.
\end{align}
%%%%%%%%%
\end{theorem}
%%%%%%%%%%%%%%%
\begin{proof}
We revisit this well-known identity using a non-covariant representation. Begin by applying the Hodge star once:
%%%%%%%%%
\begin{align}
* (dx^{i_1} \wedge \cdots \wedge dx^{i_p}) &= \frac{1}{q!} \sqrt{(-1)^t \det g} \, g^{-1}_{i_1 k_1} \cdots g^{-1}_{i_p k_p} \, \varepsilon_{j_1 \cdots j_q k_1 \cdots k_p} \, dx^{j_1} \wedge \cdots \wedge dx^{j_q}.
\end{align}
%%%%%%%%%
Applying the Hodge star again gives:
%%%%%%%%%
\begin{align}
**(dx^{i_1} \wedge \cdots \wedge dx^{i_p}) 
&= \frac{1}{q! p!} (-1)^t \det g \, g^{-1}_{i_1 k_1} \cdots g^{-1}_{i_p k_p} \, \varepsilon_{j_1 \cdots j_q k_1 \cdots k_p} \nn\\
&\quad \times g^{-1}_{j_1 s_1} \cdots g^{-1}_{j_q s_q} \, \varepsilon_{r_1 \cdots r_p s_1 \cdots s_q} \, dx^{r_1} \wedge \cdots \wedge dx^{r_p}.
\label{startstar}
\end{align}
%%%%%%%%%
Now, applying Theorem~6 (Eq.~\eqref{det.A.delta}) to contract the epsilon symbols, we obtain:
%%%%%%%%%
\begin{align}
**(dx^{i_1} \wedge \cdots \wedge dx^{i_p})
&= (-1)^{pq + t} \frac{1}{p! \, q!} \det g \cdot \det g^{-1} \cdot p! \, q! \, \delta_{r_1 \cdots r_p}^{i_1 \cdots i_p} \, dx^{r_1} \wedge \cdots \wedge dx^{r_p} \nn\\
&= (-1)^{pq + t} \, dx^{i_1} \wedge \cdots \wedge dx^{i_p},
\end{align}
%%%%%%%%%
where we used $\det g \cdot \det g^{-1} = 1$ and the identity $\delta_{r_1 \cdots r_p}^{i_1 \cdots i_p} dx^{r_1} \wedge \cdots \wedge dx^{r_p} = dx^{i_1} \wedge \cdots \wedge dx^{i_p}$. This completes the proof.
\end{proof}

\bibliographystyle{jhep}
\bibliography{HodgeDualSphere}
\end{document}